\documentclass[11pt, A4paper]{article}
\usepackage[a4paper, margin=35mm]{geometry}
\pdfoutput=1

\usepackage[utf8]{inputenc}

\usepackage{natbib}

\usepackage{microtype}
\DisableLigatures[f]{encoding = *, family = * }

\usepackage{changepage}

\usepackage[aboveskip=1pt,labelfont=bf,labelsep=period,singlelinecheck=off]{caption}

\makeatletter
\renewcommand{\@biblabel}[1]{\quad#1.}
\makeatother

\usepackage{lastpage,fancyhdr,graphicx}
\usepackage{color}

\definecolor{Gray}{gray}{.25}

\usepackage{graphicx}
\usepackage{subcaption}
\usepackage{placeins}

\usepackage{sidecap}

\usepackage{wrapfig}
\usepackage[pscoord]{eso-pic}
\usepackage[fulladjust]{marginnote}
\reversemarginpar
\usepackage{float}

\usepackage{amsmath}
\usepackage{optidef}

\usepackage{hyperref}

\usepackage{booktabs,dcolumn,rotating}
\newcolumntype{d}{D{.}{.}{2.5}}           

\usepackage{amsthm}
\newtheorem{theorem}{Theorem}
\newtheorem{lemma}{Lemma}

\begin{document}
\vspace*{0.35in}

\begin{flushleft}
{\Large
\textbf\newline{A Scalable MCEM Estimator for Spatio-Temporal Autoregressive Models}
}
\newline
\\
Philipp Hunziker,\textsuperscript{1*}
Julian Wucherpfennig,\textsuperscript{2}
Aya Kachi,\textsuperscript{3}
Nils-Christian Bormann\textsuperscript{4}
\\
\bigskip
\bf{1} Northeastern University
\\
\bf{2} Hertie School of Governance
\\
\bf{3} University of Basel
\\
\bf{4} University of Exeter

\bigskip
* Corresponding author: \href{mailto:hunzikp@gmail.com}{hunzikp@gmail.com}

\bigskip

\end{flushleft}

\section*{Abstract}
Very large spatio-temporal lattice data are becoming increasingly common across a variety of disciplines. However, estimating interdependence across space and time in large areal datasets remains challenging, as existing approaches are often (i) not scalable, (ii) designed for conditionally Gaussian outcome data, or (iii) are limited to cross-sectional and univariate outcomes. This paper proposes an MCEM estimation strategy for a family of latent-Gaussian multivariate spatio-temporal models that addresses these issues. The proposed estimator is applicable to a wide range of non-Gaussian outcomes, and implementations for binary and count outcomes are discussed explicitly. The methodology is illustrated on simulated data, as well as on weekly data of IS-related events in Syrian districts.

\section{Introduction}

Over the last decade, there has been considerable growth in the availability of spatially indexed data. This is especially the case for event-type data, which record the time and place of a specific type of occurrence (i.e. an `event'), and have become increasingly common across a range of subjects, including political violence, crime, epidemiology, traffic, geo-tagged social media data, cell-phone activity, and more. 

A common and natural approach for analyzing such data is to aggregate them onto a spatial lattice, that is, to map the events onto distinct tiles fully covering the study area.\footnote{An alternative approach which we do not address in this paper is to treat event-type data as a spatial point process, see e.g. \citealt{baddeley2015spatial}.}
In lattice data, it is often of interest to estimate the presence and magnitude of spatial (inter-)dependence, i.e. the degree to which the outcome of unit $i$ affects the outcomes of its neighbors and vice-versa. 
Modeling such spatial interdependence is at the core of much of the spatial econometrics literature, with the spatial autoregressive (SAR) model being the workhorse tool for that purpose (see \citealt{lesage2009introduction}; \citealt{anselin2010thirty}). In its simplest and most widely applied form, the SAR model is designed for a conditionally normal outcome variable observed over a cross-section. However, empirical applications rarely meet this requirement. For one, many outcomes of interest are discrete. Importantly, this is near-always the case if the outcome is constructed by aggregating event-type data onto a lattice, which will typically result in a count or binary outcome. Moreover, as recording spatially indexed data becomes easier, researchers increasingly have access to time-variant data. In this setting, it is often desirable to structure the data as a panel, and to model temporal \emph{and} spatial dependence simultaneously. And finally, many interesting problems involve more than one outcome, and researchers may be interested in studying their interaction while taking into account spatial dependence. 

Extending the simple SAR model to address these complications is non-trivial.
Perhaps most importantly, the standard estimation strategy for the normal SAR model -- maximum likelihood estimation (MLE) -- does not travel well to settings with non-normal outcomes because the respective log-likelihood functions involve a high-dimensional integrals that cannot be evaluated directly.
In addition, implementing an estimation strategy for the SAR model that scales well with large datasets is challenging, regardless of the outcome distribution. Naive implementations often rely on conceptually simple but computationally costly matrix operations which make estimating large models impractical (see e.g. \citealt{bivand2013computing}). 

This paper addresses these challenges by proposing a Monte Carlo Expectation Maximization (MCEM) estimation strategy for a class of multivariate spatio-temporal autoregressive models for non-Gaussian outcome variables. 
Specifically, we introduce an MCEM estimator for a family of models where a multivariate spatio-temporal autoregressive process operates on a latent variable with a conditionally Gaussian distribution; the observed outcomes are then assumed to be iid distributed conditional on the (transformed) latent variable. This setup permits the researcher to freely choose an appropriate outcome distribution for the type of data being modeled, such as a Poisson distribution for count data. 

The proposed estimator features a number of benefits. First, it is applicable to a large class of non-Gaussian outcomes. We discuss and provide implementations for binary and count data, but extending the estimator to other types of outcomes (e.g. censored or multinomial) is straightforward. 
Second, we take appropriate steps to ensure that our estimator is scalable, and thus applicable to large datasets. 
For univariate outcomes and settings where each unit has a constant number of spatial neighbors, our proposed implementation of the expected log-likelihood function has linear time complexity in the number of units and time-periods. For models with multivariate outcomes, our implementation has approximately quadratic time complexity in the number of units, and linear time complexity in the number of periods.
Finally, our MCEM estimator combines the benefits of a purely Bayesian (estimated e.g. via MCMC) and a purely frequentist (estimated e.g. via MLE or GMM) approach. Specifically, relying on Monte Carlo sampling during the $E$ step of the algorithm enables us to seamlessly impute missing values in the outcome in a fully Bayesian manner. This also permits generating out-of-sample predictions during model training. Simultaneously, the $M$ step of the algorithm ensures that our estimator converges much more quickly -- and with less need for tuning -- than a fully Bayesian approach.

This paper adds to the quickly growing literature that derives scalable estimators for modeling spatial dependence in non-Gaussian outcomes, specifically \cite{wang2013partial}, \cite{pace2016fast}, and \cite{liesenfeld2016likelihood}. Apart from introducing a new estimation strategy, this paper supplements these previous efforts in at least three aspects. First, we provide a model that is suitable not only for spatio-temporal data, but also for multivariate outcomes. Second, our approach provides an efficient and theoretically satisfying solution to missing outcome data and out-of-sample prediction tasks. And third, we provide an efficient, tested 'black-box' implementation of the proposed estimator in the R programming language, making it easily accessible for applied researchers.\footnote{See \url{https://github.com/hunzikp/dimstar}.}

The paper proceeds by providing an overview of the related literature. In Section \ref{sec:model}, we then introduce the family of models we are interested in, followed by a detailed discussion of our MCEM estimation strategy in Section~\ref{sec:estim}. Next, Section~\ref{sec:sim} evaluates our estimator on simulated data, showing favorable performance in terms of finite sample properties and scalability. Section~\ref{sec:application} illustrates the methodology on observed data, presenting an analysis of weekly IS-related incidences of political violence in Syrian districts. 

\section{Related Literature}

The literature on spatial autoregressive models for non-Gaussian outcomes is largely divided along outcome types. 
For one, a considerable literature discusses SAR models for binary data, in particular cross-sectional Spatial Probit models, where the spatial autoregressive process operates on a latent Gaussian variable and is mapped onto the binary outcome via the typical Probit-style indicator function (see \citealt{lesage2009introduction} for a textbook treatment).
Perhaps the best-known estimator for the Spatial Probit model is introduced by \cite{lesage2000bayesian}, who derives a Metropolis-Hastings-within-Gibbs sampling approach for estimating a fully Bayesian variant of the model for cross-sectional data. 
Taking an alternative approach, \cite{beron2004probit} propose a maximum simulated likelihood estimation strategy where the likelihood is approximated via recursive importance sampling (RIS). \cite{franzese2016spatial} generalize the RIS approach to settings with spatial and temporal dependence. 
Yet another estimation strategy is pursued by \cite{pinkse1998contracting}, who derive a GMM estimator for the Spatial Probit model. This approach is extended by \cite{klier2008clustering}, who introduce a scalable but inconsistent linearized version of \citeauthor{pinkse1998contracting}'s \citeyear{pinkse1998contracting} estimator.\footnote{More precisely, the Klier and McMillen linearized estimator produces inconsistent estimates for the spatial autoregressive parameter if absolute spatial dependence is high; see also the results reported by \cite{calabrese2014estimators}.} 

Finally, an early approach that deserves special attention in the context of this paper is the estimator proposed by \cite{mcmillen1992probit}. Similar to the present paper, McMillen considers the EM algorithm to perform maximum likelihood estimation, though he limits his focus to univariate cross-sectional binary outcome data. Within these constraints, he derives a \emph{closed-form} EM estimator, that is, unlike the approach put forward in this paper, there is no need for Monte Carlo simulation. However, this comes at the cost of consistency. Specifically, McMillen makes a number of simplifying assumptions such that the final EM estimator does not maximize a valid likelihood. As a result, the McMillen estimator produces biased estimates, as shown in the simulation experiments of \cite{calabrese2014estimators}.

Because the estimation approaches discussed so far do not generally scale well to large datasets, recent efforts in the literature have focused on deriving cross-sectional Spatial Probit estimators with low computational overhead. \cite{wang2013partial} propose a partial maximum likelihood estimator that simplifies the estimation problem by grouping observations into pairs. Similarly, \cite{pace2016fast} improve scalability by proposing a variant of the RIS approach by \cite{beron2004probit} that leverages sparse matrix methods. 

Another class of non-Gaussian spatial autoregressive models having received attention in the literature are SAR models for count data.
In an excellent review, \cite{glaser2017review}, following the time-series terminology of \cite{cox1981statistical}, divides SAR count data models into two types: 'Observation-driven' models where spatial autocorrelation operates on the observed count outcome, and 'parameter-driven' models, where spatial autocorrelation operates on a latent variable. The estimator introduced in the present paper falls into the latter category.

\cite{lambert2010two} introduce an observation-driven model for cross-sectional data where the spatial autoregressive process operates on the log-conditional-mean of a Poisson outcome. They propose an estimation strategy based on two-step limited information maximum likelihood. Conceptually similar models are introduced by \cite{hays2009comparison} and \cite{bhati2008generalized}. Apart from their focus on cross-sectional data, these models only account for spatial dependence induced by the predictors, but not by the error terms \citep[p.9]{glaser2017review}.

Among parameter-driven SAR count models, two deserve special attention.\footnote{See \cite{glaser2017review} for a more complete overview of the literature.} The first is the cross-sectional model considered by \cite{bhat2014spatial}, which, similar to the family of models considered in this paper, assumes a latent-Gaussian structure. Moreover, \cite{bhat2014spatial} also explicitly consider multivariate outcomes. However, in contrast to the approach pursued in the present paper, their formulation restricts outcome dependence to the error terms, whereas we allow for fully interdependent outcomes. 

The second SAR count model worth highlighting is the one considered by \cite{liesenfeld2016likelihood, liesenfeld2017likelihood}, who also focus on  the latent-Gaussian setup. Moreover, similar to the present paper, \cite{liesenfeld2016likelihood} do not focus exclusively on a single outcome type, but generalize their analysis to a wide range of outcome distributions, discussing count and binary outcomes specifically. Their key contribution is a new maximum simulated likelihood estimation procedure that extends the efficient importance sampling (EIS) algorithm by \cite{richard2007efficient} in a manner that makes it applicable to large-scale spatial models. Moreover, in \cite{liesenfeld2017likelihood}, the authors generalize the count outcome version of their model to a panel setup, and consider temporal and spatial dependence. 

In sum, despite being a very active strand of research, the literature on SAR type models for non-Gaussian outcome data has only recently been able to catch up with the explosive growth in spatio-temporal data and the modeling challenges that accompany it. In particular, applied researchers currently have few choices for modeling spatio-temporal dependence in large, possibly multivariate non-Gaussian response data. Most of the existing literature only considers univariate and/or cross-sectional data, and pays limited attention to scalability. Perhaps more importantly, many of the estimation techniques proposed in the literature lack an easily accessible software implementation, thus further raising the bar for practitioners to employ these methods. 

This paper addresses these shortcomings head-on. By introducing a flexible and readily implemented MCEM algorithm for estimating spatio-temporal dependence in possibly multivariate outcome data, it complements and extends the recent efforts by \cite{pace2016fast} and \cite{liesenfeld2016likelihood}.

\section{The latent-Gaussian multivariate spatio-temporal autoregressive model}
\label{sec:model}

We propose a latent-Gaussian formulation for modeling multivariate spatio-temporal data over $j = 1, \ldots, G$ outcomes, $i = 1, \ldots, N$ units, and $t = 1, \ldots, T$ periods. Let $Y_{j,t}$ be a $N \times 1$ column vector representing the $j$th outcome at period $t$. Then the class of models we consider is summarized by the following assumptions:

\begin{align*}
    p(Y_{j,t}) & = \prod_i^N p(Y_{j,i,t} | z_{j,i,t}) \\
    z_{j,t} & = \rho_j W z_{j,t} + \gamma_j z_{j,t-1} + \sum_{k \neq j}^G \lambda_{k,j} z_{k,t} + \epsilon_{j,t} \quad \forall \ t > 1 \\
    z_{j,t} & = \rho_j W z_{j,t} + \sum_{k \neq j}^G \lambda_{k,j}  z_{k,t} + \epsilon_{j,t}  \quad \text{if} \ t = 1 \\
    \epsilon_{j,t} & \sim N(0, I \sigma^2_j),
\end{align*}

whereas $W$ is a row-standardized $N \times N$ spatial weights matrix, and $\rho_j$, $\gamma_j$, and $\lambda_{k,j}$ are parameters, each constrained to the $(-1, 1)$ interval, capturing spatial-, temporal-, and outcome-dependence, respectively. To ensure identification, we also assume that $\lambda_{k,j} = \lambda_{j,k}$. Hence, in a system with $G$ outcomes we estimate ${G\choose 2}$ $\lambda$-parameters. One may also add a (outcome-specific) matrix of predictors to to the above model ($X_j \beta_j$); we omit the term for ease of exposition.

While the above-noted model specification is intuitive, it is often more practical to work with a formulation that omits outcome- and period-wise subscripts, and is specified in reduced-form. To this end, we concatenate the latent-Gaussian vectors $z_{j,t}$ first across outcomes, then across periods, yielding the following formulation:

\begin{align}
    z & = A^{-1} \epsilon \label{eq:full_model} \\
    A & = I - Q \nonumber \\
    \epsilon & \sim MVN(0, A^{-1} \Sigma A'^{-1}). \nonumber
\end{align}

$Q$ is a $NGT \times NGT$ interdependence matrix taking the form

\begin{equation*}
Q = 
\begin{pmatrix}
  Q^{*} & 0 & 0 & \cdots & 0 \\
  L^{*} & Q^{*} & 0 & \cdots & 0 \\
  0 & L^{*} & Q^{*} & \cdots & 0 \\
  \vdots  & \vdots & \ddots  & \ddots & \vdots  \\
  0 & 0 & \cdots & L^{*} & Q^{*}
\end{pmatrix},
\end{equation*}

whereas

\begin{equation*}
Q^{*}_{NG \times NG} = 
\begin{pmatrix}
  \rho_1 W & I \lambda_{1,2}  & \cdots & I \lambda_{1,G}  \\
  I \lambda_{1,2}  & \rho_2 W & \cdots & I \lambda_{2,G}  \\
  \vdots  & \vdots & \ddots  & \vdots  \\
  I \lambda_{1,G}  & I\lambda_{2,G}  & \cdots & \rho_G W
\end{pmatrix},
\end{equation*}

and

\begin{equation*}
L^{*}_{NG \times NG} = 
\begin{pmatrix}
  I \gamma_1 & 0  & \cdots & 0  \\
  0  & I \gamma_2 & \cdots & I 0  \\
  \vdots  & \vdots & \ddots  & \vdots  \\
  0  & 0  & \cdots & I \gamma_G
\end{pmatrix}.
\end{equation*}

$MVN$ denotes the multivariate normal distribution, and  

\begin{align*}
\Sigma = 
\begin{pmatrix}
  \Sigma^{*} & 0 & 0 & \cdots & 0 \\
  0 & \Sigma^{*} & 0 & \cdots & 0 \\
  0 & 0 &\Sigma^{*} & \cdots & 0 \\
  \vdots  & \vdots & \ddots  & \ddots & \vdots  \\
  0 & 0 & \cdots & 0 & \Sigma^{*}
\end{pmatrix},
\end{align*}

with 

\begin{equation*}
    \Sigma^{*}_{NG \times NG} = 
\begin{pmatrix}
  I_N \sigma^2_1 & 0 & \cdots & 0 \\
  0 & I_N \sigma^2_2 & \cdots & 0 \\
  \vdots  & \vdots & \ddots  & \vdots  \\
  0 & 0 & \cdots & I_N \sigma^2_G
\end{pmatrix}.
\end{equation*}

$I_N$ denotes the $N \times N$ identity matrix.
Note that if $G=1$, the proposed model simplifies to the regular spatio-temporal autoregressive model (STAR) defined on the latent-Gaussian outcome $z$. The proposed model is stationary if $max_i |\sum_j Q_{i,j}| < 1$, which we ensure by requiring

\begin{equation*}
    \left |\sum_{(j,k)} \lambda_{j,k} + \rho_j + \gamma_j \right| < 1 \quad \forall j,
\end{equation*}

whereas $\sum_{(j,k)}$ represents the sum over all possible pairs of outcomes involving outcome $j$.

We link the latent-Gaussian outcome $z$ to the observed outcome $Y$ by specifying an appropriate distribution $p(Y_{j,i,t} | z_{j,i,t})$. If the observed outcome is binary we propose the deterministic distribution function

\begin{equation*}
    p(Y_{j,i,t} = 1 | z_{j,i,t}) = 
    \begin{cases}
        1 ,& \text{if } z_{j,i,t} \geq 0 \\
        0,              & \text{otherwise},
    \end{cases}
\end{equation*}

which, together with the identifying assumption $\sigma^2_j = \sigma^2 = 1$, amounts to a multivariate spatio-temporal Probit specification.

If the observed outcome is a count variable, we propose using a multivariate Poisson log-Normal setup, i.e.

\begin{equation*}
    p(Y_{j,i,t} | z_{j,i,t}) = \text{dPois}(Y_{j,i,t}; \ exp(z_{j,i,t}))
\end{equation*}

whereas $\text{dPois}$ is the Poisson PMF.

\section{Estimation}
\label{sec:estim}

Let $\theta$ denote a vector of all parameters ($\theta = {\rho_j, \gamma_j, \lambda_{k,j}}, \sigma^2_j, \ldots$). The full-data log-likelihood of the model is then given by

\begin{align}
    ln \ L(\theta | Y, z) & = ln \ p(Y, z| \theta) \nonumber \\
                         & = ln \ p(Y | z) + ln \ p (z | \theta) \nonumber \\
                         & = \sum_{i,j,t} ln \ p(Y_{j,i,t} | z_{j,i,t}) + ln \ p (z | \theta) 
                         \label{eq:fulldata}
\end{align}

The log probability of the latent-Gaussian variable $z$, in turn, is induced by the model specification in \eqref{eq:full_model}, and is given by

\begin{equation}
    ln \ p (z | \theta) = ln \ \text{dMVN}(z; \ 0, \ A^{-1} \Sigma A'^{-1}), \label{eq:zdist}
\end{equation}

whereas  $\text{dMVN}$ denotes the multivariate normal density.

\subsection*{MCEM algorithm}

The Expectation Maximization (EM; \citealt{dempster1977maximum}) formalism estimates $\theta$ by maximizing the intractable marginal likelihood

\begin{equation*}
    L(\theta | Y) = \int_{-\infty}^{\infty} p(Y, z| \theta) \ d z
\end{equation*}

via coordinate ascent, giving the following recursive algorithm:

\begin{enumerate}
    \item Determine the expected value of the full-data log-likelihood with respect to the latent variable $z$ and conditional on the current parameters $\theta^{m}$. This expectation is typically called the $Q$ function:
    
    \begin{align*}
        Q(\theta | \theta^{m}) & = E_{z | \theta^{m}, Y} \left[ ln \ p(Y, z| \theta) \right]
    \end{align*}
    
    \item Obtain the updated set of parameters by maximizing the $Q$ function:
    
    \begin{equation*}
        \theta^{m+1} = \text{argmax}_{\theta}  \ Q(\theta | \theta^{m})
    \end{equation*}
    
    \item Repeat until convergence.
\end{enumerate}

In the following, we derive the $Q$ function up to a constant. For ease of exposition, we drop all subscripts from the expectation operator.

\begin{align*}
     Q(\theta | \theta^{m}) & = E \left[ ln \ p(Y, z| \theta) \right] \\
                            & = E \left[ ln \ p(Y | z) \right] + E \left[ ln \ p(z | \theta) \right]  \\
                            & \propto E \left[ln \ p(z | \theta) \right] \\
                            & \propto ln |A| - ln {|\Sigma|}^{\frac{1}{2}} - \frac{1}{2} E \left[ (Az)' \Sigma^{-1} (Az) \right]
\end{align*}

Evaluating the expected kernel $E \left[ (Az)' \Sigma^{-1} (Az) \right]$ requires the first two moments of $z | \theta^{m}, Y$, which are generally not available in closed form. For this reason, we pursue an MCEM approach where we obtain an estimate of the kernel via sampling. One way to do this would be to (i) sample from $p(z | \theta^{m}, Y)$, (ii) use these samples to estimate the first two moments of $z | \theta^{m}, Y$ and then (iii) plugging these estimates into the $Q$ function. However, this approach requires estimating the $NGT \times NGT$ covariance matrix $V \left[z |  \theta^{m}, Y \right]$, which involves considerable computational overhead. Instead, we invoke the law of the unconscious statistician and estimate the kernel directly via the numerical integral

\begin{equation*}
    E \left[ (Az)' \Sigma^{-1} (Az) \right] \approx \frac{1}{S} \sum_{s=1}^S \left[ (Az^{(s)})' \Sigma^{-1} (Az^{(s)}) \right].
\end{equation*}

\subsection*{MC sampling}

The MCEM algorithm requires sampling from 

\begin{align*}
    p(z | \theta^{m}, Y) & \propto p(Y | z) \ p(z | \theta^{m}) \\
                         & \propto \prod_{i,j,t} p(Y_{i,j,t} | z_{i,j,t}) \ p(z | \theta^{m}).
\end{align*}

Sampling from this distribution directly is typically either unfeasible (e.g. in the count-data case) or prohibitively costly (e.g. in the binary case), which is why we propose using a Gibbs sampler to sample from

\begin{align}
    p(z_{l} | \theta^{m}, Y, z_{-l}) & \propto p(Y_{l} | z_{l}) \ p(z_{l} | \theta^{m}, z_{-l}) \label{eq:gibbs},
\end{align}

whereas subscript $l$ enumerates a single data point, and $z_{-l}$ denotes the entire $z$ vector except for $l$. Note that the second term in expression \eqref{eq:gibbs} is a conditional multivariate normal distribution, and thus itself normal:

\begin{align*}
    p(z_{l} | \theta^{m}, z_{-l}) & = \text{dN}(z_{l}; \ \bar{\mu}, \bar{\Sigma}) \\ 
    \bar{\mu} & = H_{ii}^{-1} H_{l, -l} z_l  \\ 
    \bar{\Sigma} & = H_{ii}^{-1},
\end{align*}

where $H = A' \Sigma^{-1} A$, i.e. $H$ is the precision matrix of $z$, as established in \eqref{eq:zdist}. Note that a fortunate implication of these derivations is that we do \emph{not} have to invert $A$ to evaluate $p(z_l | \cdot)$.  We \emph{do} have to invert $\Sigma$, but because $\Sigma$ is diagonal, this may be performed in linear time. Similarly, $H_{ii}$ is scalar, so obtaining its inverse is straightforward as well.

If the outcome is binary, sampling from the conditional distribution \eqref{eq:gibbs} thus amounts to sampling from the univariate truncated normal distribution,

\begin{equation*}
    p(z_{l} | \theta^{m}, Y, z_{-l}) = 
    \begin{cases}
        \text{dTN}_{0, \infty}(z_l; \ \bar{\mu}, \bar{\Sigma}) ,& \text{if } Y_l = 1 \\
        \text{dTN}_{-\infty, 0}(z_l; \ \bar{\mu}, \bar{\Sigma}),              & \text{otherwise},
    \end{cases}
\end{equation*}

where $\text{dTN}_{a, b}$ denotes the truncated Normal density with support in the interval $(a, b]$. 

If the outcome is a count-variable, sampling from the conditional distribution \eqref{eq:gibbs} requires sampling from the unnormalized Poisson log-Normal posterior
\begin{align*}
    p(z_{l} | \theta^{m}, Y, z_{-l}) & \propto p(Y_{l} | z_{l}) \ \text{dN}(z_{l}; \ \bar{\mu}, \bar{\Sigma}) \\ 
                                     & \propto \text{dPois}(Y_l; \ exp(z_l)) \ \text{dN}(z_{l}; \ \bar{\mu}, \bar{\Sigma}).
\end{align*}

To do so, we employ the adaptive rejection sampler by \cite{wild1993algorithm}, which allows highly efficient sapling from unnormalized densities as long as (i) the first derivative with respect to the sampling variable is available, and (ii) the density is log-concave. For the present application, it is trivial to show that these requirements are met. The logarithm of the Poisson log-Normal posterior is proportional to

\begin{align*}
    ln \ p(z_{l} | \theta^{m}, Y, z_{-l}) & \propto Y_l z_l - exp(z_l) - 0.5\bar{\Sigma}^{-1}(z_l - \bar{\mu})^2
\end{align*}

with first and second derivatives

\begin{align*}
    \frac{\partial ln \ p(z_{l} | \cdot)}{\partial z_l} & = Y_l - exp(z_l) - \bar{\Sigma}^{-1}(z_l - \bar{\mu}) \\ 
    \frac{\partial ln \ p(z_{l} | \cdot)}{\partial^2 z_l} & = -exp(z_l) - \bar{\Sigma}^{-1} < 0.
\end{align*}

Finally, note that using a Gibbs sampler to sample the latent variable $z$ provides a natural strategy for imputing missing data. Specifically, if $Y_{i,t,l} = Y_l$ is missing, we simply draw $z_l$ from its conditional prior,

\begin{align*}
    p(z_{l} | \theta^{m}, z_{-l}, Y, Y_l = \text{missing}) & = p(z_{l} | \ \theta^{m}, z_{-l}) \\ 
    & = \text{dN}(z_{l}; \ \bar{\mu}, \ \bar{\Sigma}),
\end{align*}

during the $E$ step of the MCEM algorithm. This yields a complete set of $z$ samples, allowing us to proceed with the $M$ step without further modifications. This approach is particularly useful for generating out-of-sample predictions or forecasts. We simply add the observations to be predicted to the dataset on which the model is being trained, and set corresponding outcome observations to missing. The MCEM algorithm will then produce corresponding $z$ samples from which predictions can be obtained. 

\subsection*{Efficient evaluation of the Q function}

A central part of the MCEM algorithm is the maximization of the $Q$ function, given by

\begin{equation*}
     Q(\theta | \theta^{m}) \propto ln |A| - ln {|\Sigma|}^{\frac{1}{2}} - 0.5 \underbrace{\frac{1}{S} \sum_{s=1}^S \left[ (Az^{(s)})' \Sigma^{-1} (Az^{(s)}) \right]}_\text{approx. expected kernel}
\end{equation*}

Because the $Q$ function is evaluated repeatedly during optimization, scalability requires that $Q$ can be evaluated efficiently even for very large datasets. In the following, we briefly discuss our computational setup for evaluating the three constituent terms of the $Q$ function as efficiently as possible.

\paragraph{Expected kernel}
We ensure efficient evaluation of the approximate expected kernel by storing $A$ as a sparse matrix, and using sparse matrix multiplication for evaluating each kernel sample. This yields a complexity that is linear in $nnz \times S$, whereas $nnz$ is the number of non-zero entries in $A$.\footnote{Because $\Sigma^{-1}$ is diagonal, the matrix product $v' \Sigma^{-1} v$ with $v_{NGT \times 1}$ has complexity $\mathcal{O}(N*G*T)$.} If $W$ is such that each unit has $K$ spatial neighbors, then 

\begin{equation*}
    nnz = \underbrace{T*N*G}_{\text{diag.}} + \underbrace{(T-1)*N*G}_{\text{temp. lags}} + \underbrace{T*G*N*K}_{\text{spat. lags}} + \underbrace{T*G*(G-1)*N}_{\text{outcome lags}}
\end{equation*}

Hence, assuming $K$ constant as $N$ increases (which is typically the case in lattice data), then the complexity of evaluating the approximate expected kernel is $\mathcal{O}(S*T*G^2*N)$. In other words, it is linear in the number of units, time-periods, and MC samples, and quadratic in the number of outcomes.

\paragraph{Log-determinant of error variance} Recall that the error covariance matrix, $\Sigma$, is assumed diagonal. It follows that 

\begin{equation*}
ln {|\Sigma|}^{\frac{1}{2}} = \sum_{l}^{NGT} ln \Sigma_{ll}^\frac{1}{2},
\end{equation*}

which has complexity $\mathcal{O}(N*G*T)$.

\paragraph{Log-determinant of $A$} The biggest impediment towards an efficient implementation of the $Q$ function is the computation of the log-determinant of the spatio-temporal multiplier, $ln |A|$. Standard algorithms for computing the determinant of arbitrary martrices have cubic complexity, which is clearly prohibitive when $N*G*T$ is large. There is an extensive literature discussing this issue for the simple spatial autoregressive case where $A = I - \rho W$, see e.g. \cite{bivand2013computing}. In our case, however, $A$ also contains temporal and outcome-specific lags, making many of the standard approaches for evaluating the log-determinant -- e.g. \citeauthor{ord75estimation}'s (\citeyear{ord75estimation}) method -- impractical.

We address this issue by showing that, thanks to the particular structure of $A$, it is actually not necessary to compute the log-determinant of $A$ directly. More precisely, in the following, we introduce what we term the \emph{STAR Log-Determinant} theorem, which establishes that computing the log-determinant of the $NGT \times NGT$ matrix $A$ only requires numerical evaluation of the determinant of a smaller matrix of size $NG \times NG$.

We start by stating the following results:

\begin{lemma}
The inverse of a block-diagonal matrix is the block-diagonal matrix of block inverses. E.g.
\begin{equation*}
    B^{-1} = 
\begin{pmatrix}
  B_1 & 0 \\
  0 & B_2 \\
\end{pmatrix}^{-1}
= 
\begin{pmatrix}
  B_1^{-1} & 0 \\
  0 & B_2^{-1} \\
\end{pmatrix}
\end{equation*}
\label{lem:block}
\end{lemma}

\begin{lemma}
The left or right matrix product of a block-diagonal matrix $B$ with a matrix $L$ that is lower-block-diagonal with respect to $B$ is itself lower-block-triangular with respect to $B$. E.g. let  
\begin{equation*}
    B = \begin{pmatrix}
  B_1 & 0 & 0 \\
  0 & B_2 & 0 \\
  0 & 0 & B_3
\end{pmatrix}, \quad
L = \begin{pmatrix}
  0 & 0 & 0 \\
  L_1 & 0 & 0 \\
  L_2 & L_3 & 0
\end{pmatrix}.
\end{equation*}
Then 
\begin{equation*}
    BL = \begin{pmatrix}
  0 & 0 & 0 \\
  \neq 0 & 0 & 0 \\
  \neq 0 & \neq 0 & 0
\end{pmatrix},
\end{equation*} 
and equivalently for $LB$.
\label{lem:tria}
\end{lemma}

\begin{lemma}
From Sylvester's identity \citep[see e.g.][]{akritas1996various} it follows that if $A$ is an invertible $m \times m$ matrix and $B$ is a $m \times m$ matrix, then
\begin{equation*}
    |A + B| = |A| * |I + A^{-1}B|.
\end{equation*}
\label{lem:sylv}
\end{lemma}

We now posit the following:

\begin{theorem}[STAR Log-Determinant]
Let $A$ be the $NGT \times NGT$ inverse spatio-temporal multiplier as defined in \eqref{eq:full_model}. We may decompose $A$ as follows, 
\begin{equation*}
    A = I - L - D
\end{equation*}
with 
\begin{equation*}
L = 
\begin{pmatrix}
 0 & 0 & 0 & \cdots & 0 \\
  L^{*} & 0 & 0 & \cdots & 0 \\
  0 & L^{*} & 0 & \cdots & 0 \\
  \vdots  & \vdots & \ddots  & \ddots & \vdots  \\
  0 & 0 & \cdots & L^{*} & 0
\end{pmatrix}, \quad
D = 
\begin{pmatrix}
  Q^{*} & 0 & 0 & \cdots & 0 \\
  0 & Q^{*} & 0 & \cdots & 0 \\
  0 & 0 & Q^{*} & \cdots & 0 \\
  \vdots  & \vdots & \ddots  & \ddots & \vdots  \\
  0 & 0 & \cdots & 0 & Q^{*}
\end{pmatrix},
\end{equation*}
with $L^{*}$ and $Q^{*}$ both of size $NG \times NG$. 
Then the log-determinant of $A$ is given by
\begin{equation*}
    ln |A| = T * ln |I - Q^{*}|.
\end{equation*}
\end{theorem}

\begin{proof}
Invoking Lemma \ref{lem:sylv} we can write
\begin{equation*}
    |A| = |I - D - L| = |I-D| * |I - (I -D)^{-1}L|.
\end{equation*}
From Lemmas \ref{lem:block} and \ref{lem:tria} we can derive that $(I -D)^{-1}L$ is lower-block-diagonal with respect to $I-D$. Consequently, $|I - (I -D)^{-1}L|$ is lower-triangular with a unit diagonal. Because the determinant of triangular matrices is the product of all diagonal entries, it follows that
\begin{equation*}
    |A| = |I - D - L| = |I-D|.
\end{equation*}
$|I-D|$ is block-diagonal with $T$ blocks of $I - Q^{*}$. Since the determinant of a block-diagonal matrix is the product of the block determinants, we have
\begin{equation*}
    ln |A| = T * ln |I-Q^{*}|.
\end{equation*}
\end{proof}

Given the \emph{STAR Log-Determinant} Theorem the problem of evaluating $ln |A|$ reduces to finding an efficient way to evaluate $|I-Q^{*}|$. We note that if $G = 1$ (i.e. only a single outcome is modeled), then $|I-Q^{*}| = |I-\rho W|$, and thus any of the various methods proposed for evaluating the log-Jacobian of the spatial autoregressive model may be used (see \citealt{bivand2013computing}). One popular strategy for this setting is \citeauthor{ord75estimation}'s method (\citeyear{ord75estimation}). For cases where $N$ is large and Ord's method is impractical, \cite{smirnov2009n} propose a method that has complexity $\mathcal{O}(N)$. In our implementation of the MCEM estimator, we employ the same strategy as \cite{wilhelm2013estimating} and precompute $|I-\rho W|$ for a range of $\rho$ values prior to optimization using a sparse $LU$ decomposition, and simply look up the value of the log-determinant during optimization.

If $G > 1$, we propose using a sparse Cholesky decomposition for computing $|I-Q^{*}|$, similar to the approach proposed by \cite{pace1997sparse} for the regular spatial autoregressive case. Using the sparse Cholesky decomposition raises a problem: The Cholesky decomposition is only applicable to symmetric, positive-definite matrices, but there is no guarantee of $|I-Q^{*}|$ meeting either of these requirements. Indeed, if $W$ is row-standardized, then  $|I-Q^{*}|$ will \emph{not} be symmetric. We address this problem by relying on the following `trick' suggested by \cite{pace1997sparse}:

\begin{align*}
    |I-Q^{*}| & = (|I-Q^{*}||I-Q^{*}|)^\frac{1}{2} \\
              & = (|(I-Q^{*})'||I-Q^{*}|)^\frac{1}{2} \\
              & = (|(I-Q^{*})'(I-Q^{*})|)^\frac{1}{2},
\end{align*}

whereas $(I-Q^{*})'(I-Q^{*})$ is both positive-definite and symmetric.

The computational complexity of the sparse Cholesky decomposition is not well defined, and depends on the fill-pattern of the matrix to be decomposed. However, we show in the simulations below that if $W$ is built from a regular lattice grid and the number of neighbors per unit remains constant, then the complexity is approximately $\mathcal{O}(G^2 * N^2)$.

\subsection*{Optimization}

At each MCEM iteration, we need to solve the following constrained optimization problem:

\begin{maxi*}|l|
  {\theta}{Q(\theta | \theta^{m})}{}{}
  \addConstraint{\textstyle |\sum_{(j,k)} \lambda_{j,k} + \rho_j + \gamma_j|}{<1}{}
  \addConstraint{\sigma^2_j}{>0}{}
  \addConstraint{|\rho_j|}{<1}{}
  \addConstraint{|\gamma_j|}{<1}{}
  \addConstraint{|\lambda_{j,k}|}{<1, \quad}{\forall \ j \in G \ \land \ \text{pairs}(j,k) \in G.}
\end{maxi*}

We enforce the constraints on the scalar dependence parameters ($\rho_j$, $\gamma_j$, $\lambda_{j,k}$) by optimizing with respect to the inverse hyperbolic tangent (e.g. $atanh(\rho_j)$) of the parameters, and mapping these transformed values onto the $(-1, 1)$ interval during the evaluation of the $Q$ function. Analogously, we optimize with respect to the inverse-softplus transformed variance parameters.\footnote{The softplus function (also known as the smooth linear rectifier) is given by $softplus(x) = ln(1 + exp(x))$.} The stationarity constraint $|\sum_{(j,k)} \lambda_{j,k} + \rho_j + \gamma_j|$ we enforce by employing a constrained optimization algorithm, namely the general nonlinear augmented Lagrange multiplier method solver by \cite{ye1988interior}. 

\subsection*{Standard Errors}

We estimate standard errors for $\theta$ using an MC approximation to \citeauthor{louis1982finding}' (\citeyear{louis1982finding}) method, similarly to \cite{natarajan2000monte}. 
Louis gives the following identity for the observed Fisher information matrix in the EM framework,

\begin{equation*}
    I(\theta) = -E_{z}\left[ H(\theta | Y) \right] - V_z \left[ S(\theta | Y) \right],
\end{equation*}

whereas $H$ and $S$ are the Hessian and gradient of the full-data log-likelihood function \eqref{eq:fulldata}, respectively. We obtain estimates of the expectation (and, respectively, variance) of these quantities via simulation and numerical differentiation.

\section{Simulation Results}
\label{sec:sim}

\begin{figure}[t]
\includegraphics[width=\textwidth]{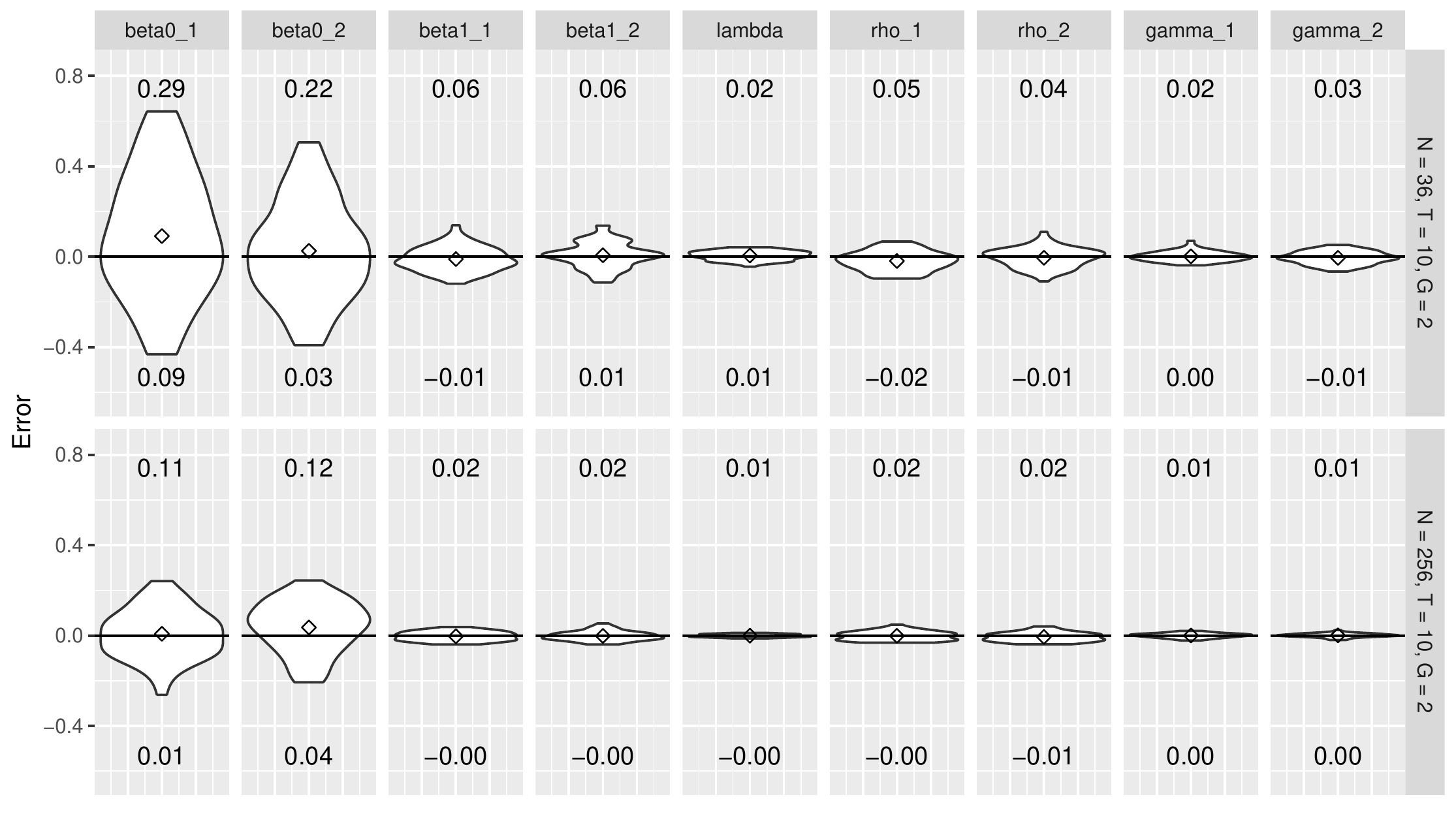}
\caption{Error distribution for all parameter estimates in a spatio-temporal count model with a bivariate outcome, for two different sample sizes. Estimates based on 50 Monte Carlo simulations per sample size. Top figures in each panel indicates RMSE, bottom figures indicate bias.}
\label{fig:bias_count}
\end{figure}

\subsection*{Bias, RMSE, and Coverage}

The goal of this first simulation study is to evaluate the finite sample properties of our MCEM estimator, and to validate our implementation of the underlying algorithms.

We generate simulated data from the following specification,
\begin{align*}
    \epsilon_{j,t} & \sim N(0, I) \\
    z_{j,t} & = X_{j,t} \beta_j + \rho_j W z_{j,t} + \gamma_j z_{j,t-1} + \sum_{k \neq j} \lambda_{k,j}  z_{k,t} + \epsilon_{j,t} \quad \forall \ t > 1 \\
    z_{j,t} & = X_{j,t} \beta_j + \rho_j W z_{j,t} + \sum_{k \neq j} \lambda_{k,j}  z_{k,t} + \epsilon_{j,t} \quad \text{if} \ t = 1 \\
    Y_{j,i,t} & \sim \text{Pois}(exp(z_{j,i,t})),
\end{align*}

whereas each $X_{j,t}$ consists of a unit constant, and a single covariate drawn from $N(0, 1)$. The true parameters are held constant at $\beta_j = [2, 1]'$ and $\rho_j = \lambda_{k,j} = \gamma_j = 0.25$. The spatial weights matrix $W$ is constructed from a regular square grid with side length $N^{0.5}$. We perform two sets of experiments, one with a small sample size of $\{G = 2, \ T = 10, \ N = 36\}$, and one with a larger sample size of $\{G = 2, \ T = 10, \ N = 256\}$. Each experiment is repeated 50 times.\footnote{All experiments reported in this subsection involve count outcomes. Results for a set of equivalent experiments with binary outcome data are available in the Appendix, with largely identical findings.}

For each simulated dataset, we use our MCEM estimator to recover point and standard error estimates for the full parameter vector $\theta$. We use a MC sample size of $S = 50$ for the MCEM algorithm, and a sample size of $S = 100$ for estimating the Fisher Information matrix. We run the MCEM algorithm for 50 iterations, unless the maximum absolute change across all estimated parameters is smaller than $1^{-4}$ in any previous iteration. 

Figure \ref{fig:bias_count} summarizes the distribution of the parameter-wise estimation error ($\hat{\theta}_k - \theta_k$) over all simulation draws, for both sets of experiments. Clearly, the MCEM estimator is able to recover the true parameter values, and the decrease in bias and RMSE (root mean square error) with larger sample size suggests consistency.  We note that the recovered estimates have low RMSE (compared to the true parameter values) even though the models were estimated with a relatively small MC sample size. Increasing the MC sample size would likely yield a further reduction in RMSE.

\begin{figure}[t]
\includegraphics[width=\textwidth]{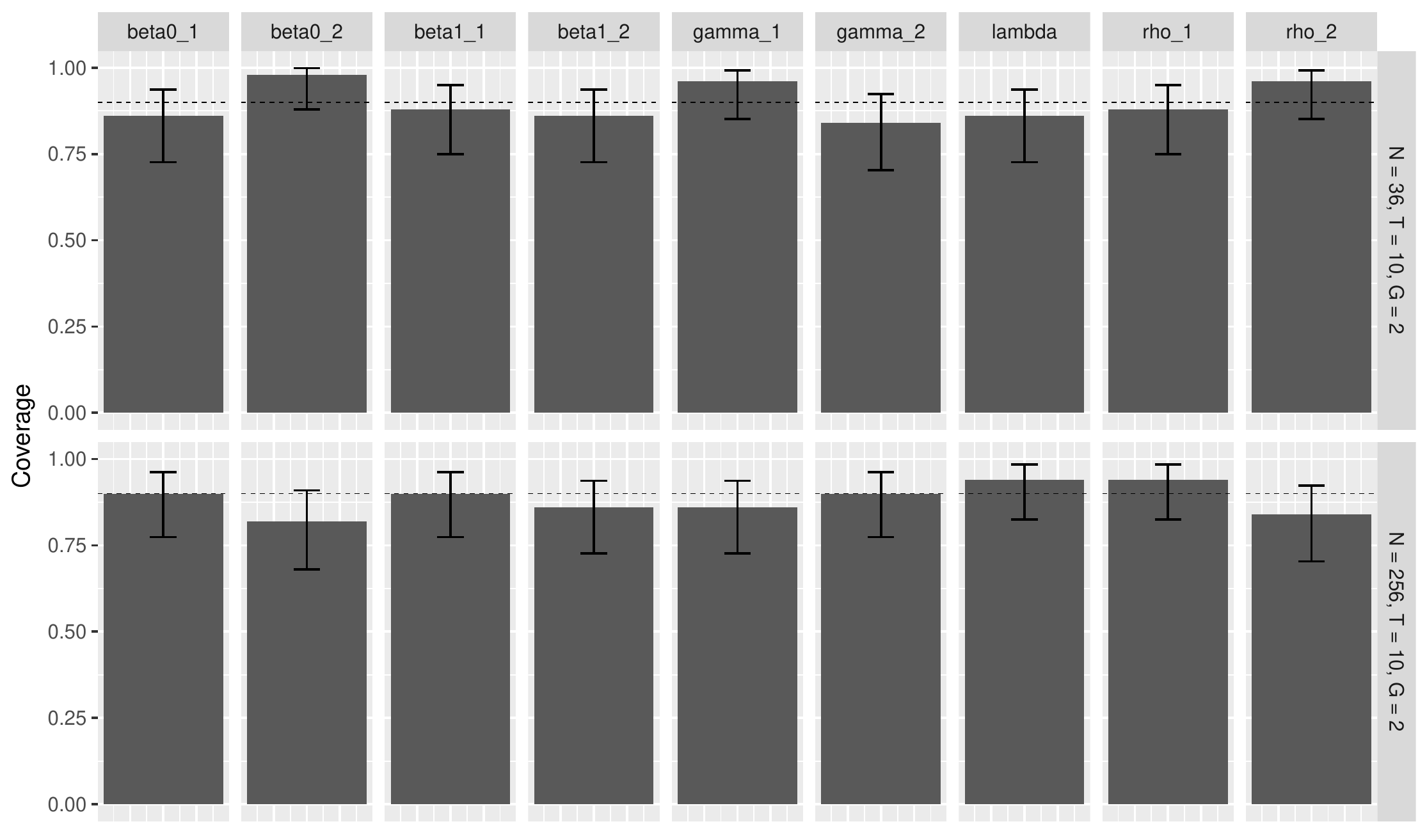}
\caption{90\% CI coverage probability for all parameter estimates in a spatio-temporal count model with a bivariate outcome. Estimates based on 50 Monte Carlo simulations per sample size. Error bars represent 95\% CI for the coverage proportion.}
\label{fig:coverage_count}
\end{figure}

Next, Figure \ref{fig:coverage_count} shows the coverage probability of the 90\% confidence intervals constructed from the estimated standard errors. Again, we observe that the standard error estimates appear to be unbiased; in no instance can we reject the null that the true coverage probability is 90\%.

\subsection*{Time Complexity of Expected Likelihood Evaluation}

\begin{figure}[t]
 
\begin{subfigure}{0.5\textwidth}
\includegraphics[width=1\linewidth]{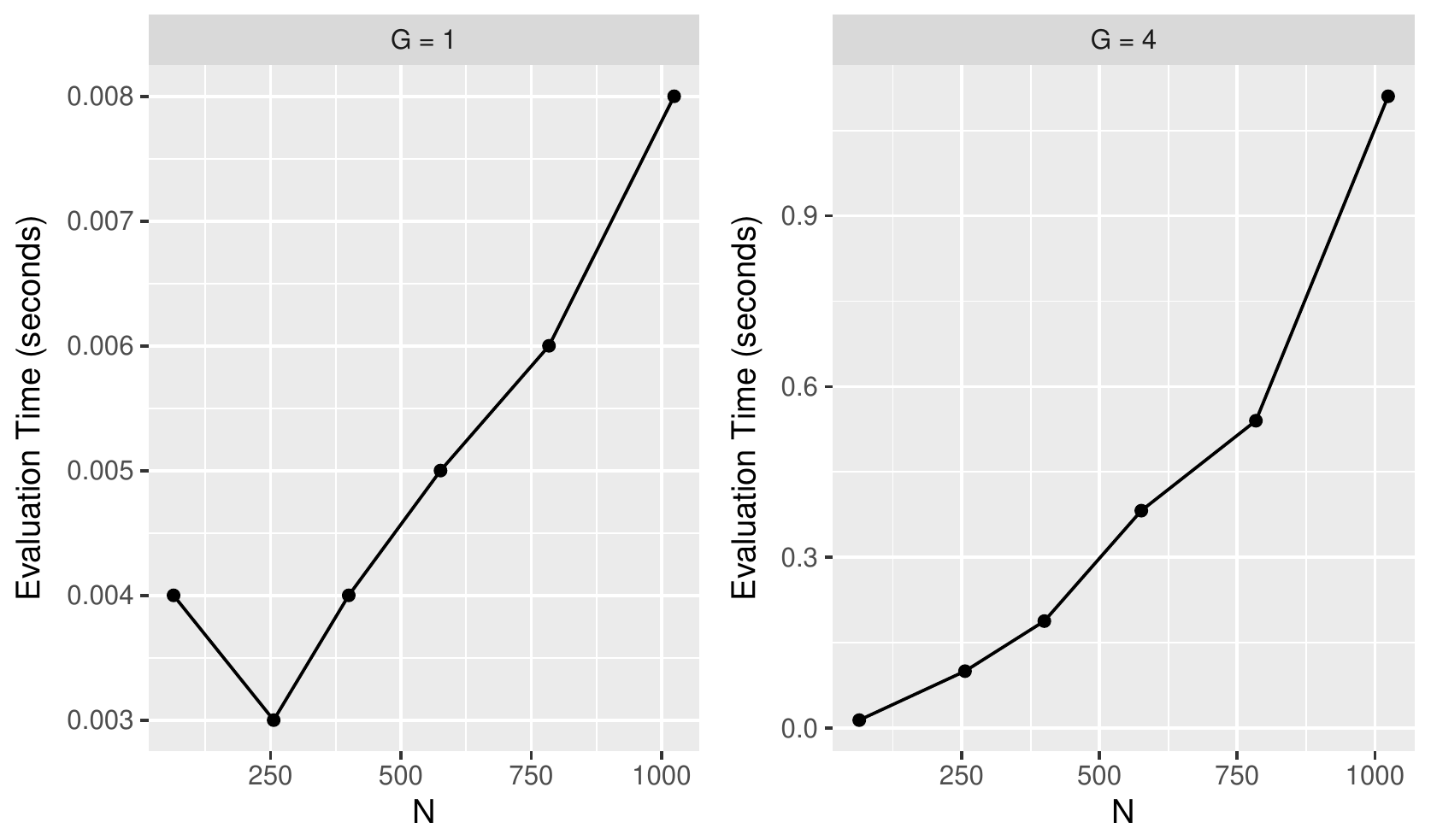} 
\caption{\vspace{0.5cm}CPU time in $N$.}
\label{fig:subim1}
\end{subfigure}
\begin{subfigure}{0.5\textwidth}
\includegraphics[width=1\linewidth]{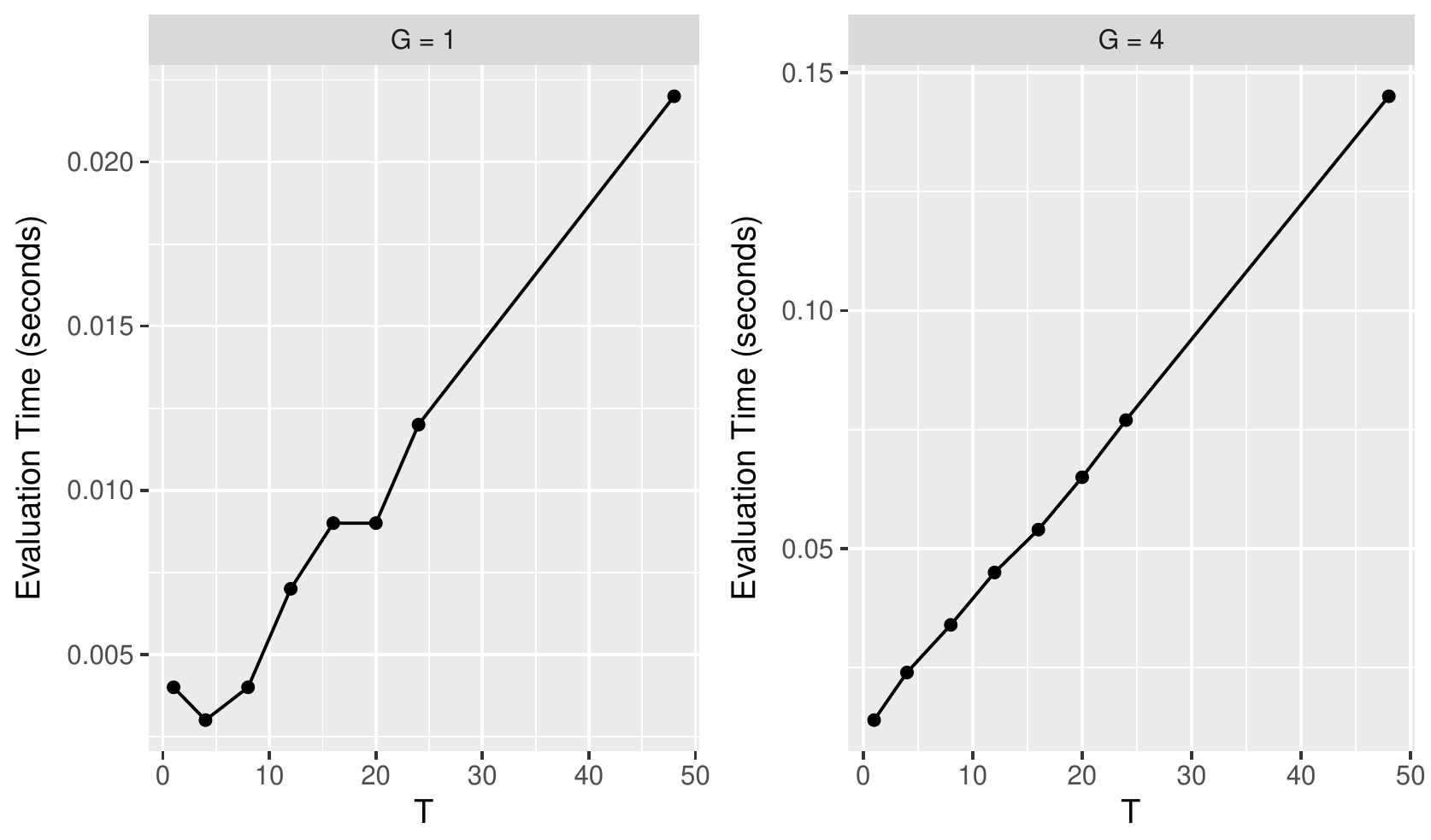}
\caption{\vspace{0.5cm}CPU time in $T$.}
\label{fig:subim2}
\end{subfigure}
\begin{subfigure}{0.5\textwidth}
\includegraphics[width=1\linewidth]{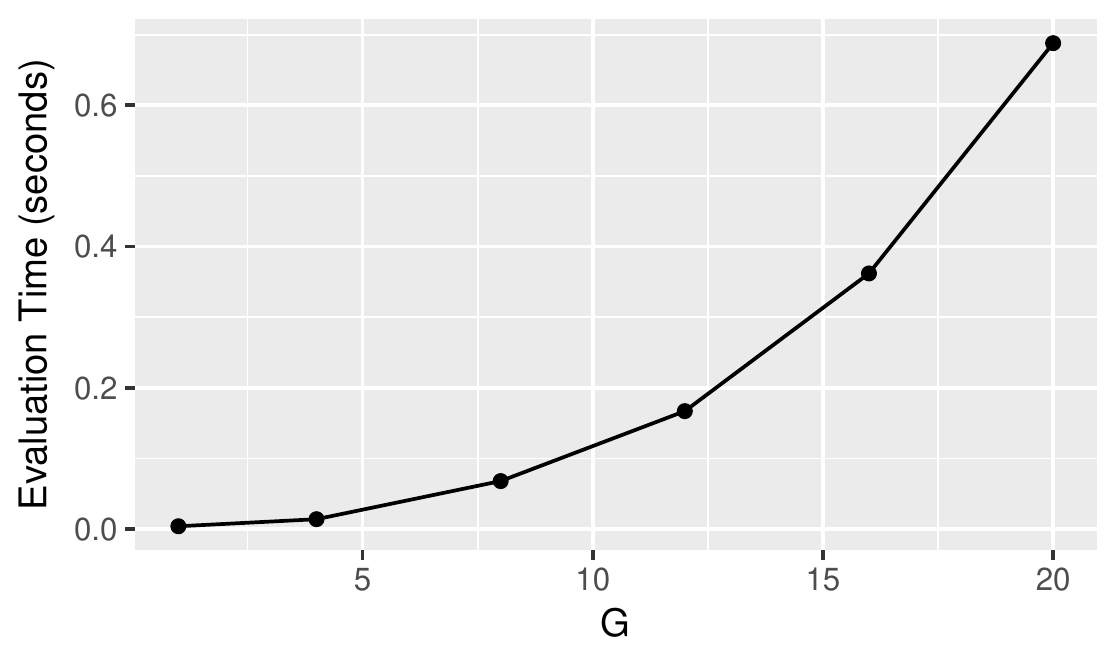}
\caption{\vspace{0.5cm}CPU time in $G$.}
\label{fig:subim3}
\end{subfigure}
\caption{Evaluation time of the expected log-likelihood function (the $Q$ function) as a function of $N$, $T$, and $G$. Each estimate based on 5 iterations.}
\label{fig:complexity}
\end{figure}

The most important step towards a scalable implementation of proposed estimator is ensuring that the evaluation of the expected log-likelihood (or $Q$ function) in the $M$ step of the MCEM algorithm is as efficient as possible. In the following set of simulation experiments, we evaluate how well our implementation of the expected log-likelihood scales in the number of units ($N$), time periods ($T$) and outcomes ($G$). This exercise serves two purposes: first, to validate the analytical complexity estimates derived in Section \ref{sec:estim}, and second, to obtain empirical complexity estimates for datasets where $G>1$. Recall that in the latter case, we were unable to derive complexity analytically because of the sparse Cholesky decomposition involved.

For all experiments, we sample $S = 50$ instances of the latent-Gaussian $z$ using the same data generating process, true parameter values, and spatial weights matrix $W$ as discussed in the previous section. Each experiment consists of evaluating the expected log-likelihood five times at the true parameter values; we report the total CPU time accumulated during those five evaluations. We repeat the experiment for different dataset sizes, varying $N$, $T$ and $Q$ separately. 

Figure \ref{fig:complexity} summarizes the results. Panel \ref{fig:subim1} shows CPU time as a function of $N$, for $G = \{1, 4\}$ and $T = 1$. As expected, performance is linear in $N$ for the univariate outcome case ($G = 1$), except for very small $N$. When $G = 4$, we find that performance is approximately quadratic. As discussed extensively in Section \ref{sec:estim}, this is because if $G > 1$, we need to perform a sparse Cholesky decomposition of a matrix of size $NG \times NG$ to evaluate $ln |A|$. While the complexity of this operation is problem dependent, the results presented here suggest that it is much more efficient than the regular (dense) Cholesky decomposition, which has cubic complexity.

Next, Panel \ref{fig:subim2} shows CPU time as a function of $T$, with $G = \{1, 4\}$ and $T = 1$. As expected, we find that performance is linear in $T$, even for $G>1$. The latter property is thanks to the \emph{STAR Log-Determinant} theorem, which permits efficient evaluation of $ln |A|$ regardless of $T$.

Finally, Panel \ref{fig:subim3} shows changes in CPU time as we vary $G$, with $N = 64$ and $T = 1$. As anticipated in Section \ref{sec:estim}, we find that performance is quadratic in $G$.

In summary, the cost of evaluating the expected log-likelihood is linear in the number of units and time-periods if the outcome is univariate, and roughly quadratic in the number of units if the outcome is multivariate. Also, evaluating the expected log-likelihood becomes quadratically more expensive as the number of modeled outcomes increases. While these findings do not translate into total run-times for the full MCEM algorithm directly (which depends on the problem at hand and the convergence criterion employed), they demonstrate that our estimator is far more efficient than a `naive` approach that relies on direct evaluation of the log-determinant and dense matrix algebra.

\subsection*{Spatial Probit Comparison}
\label{sec:spat_probit}

\begin{figure}[t]
\includegraphics[width=\textwidth]{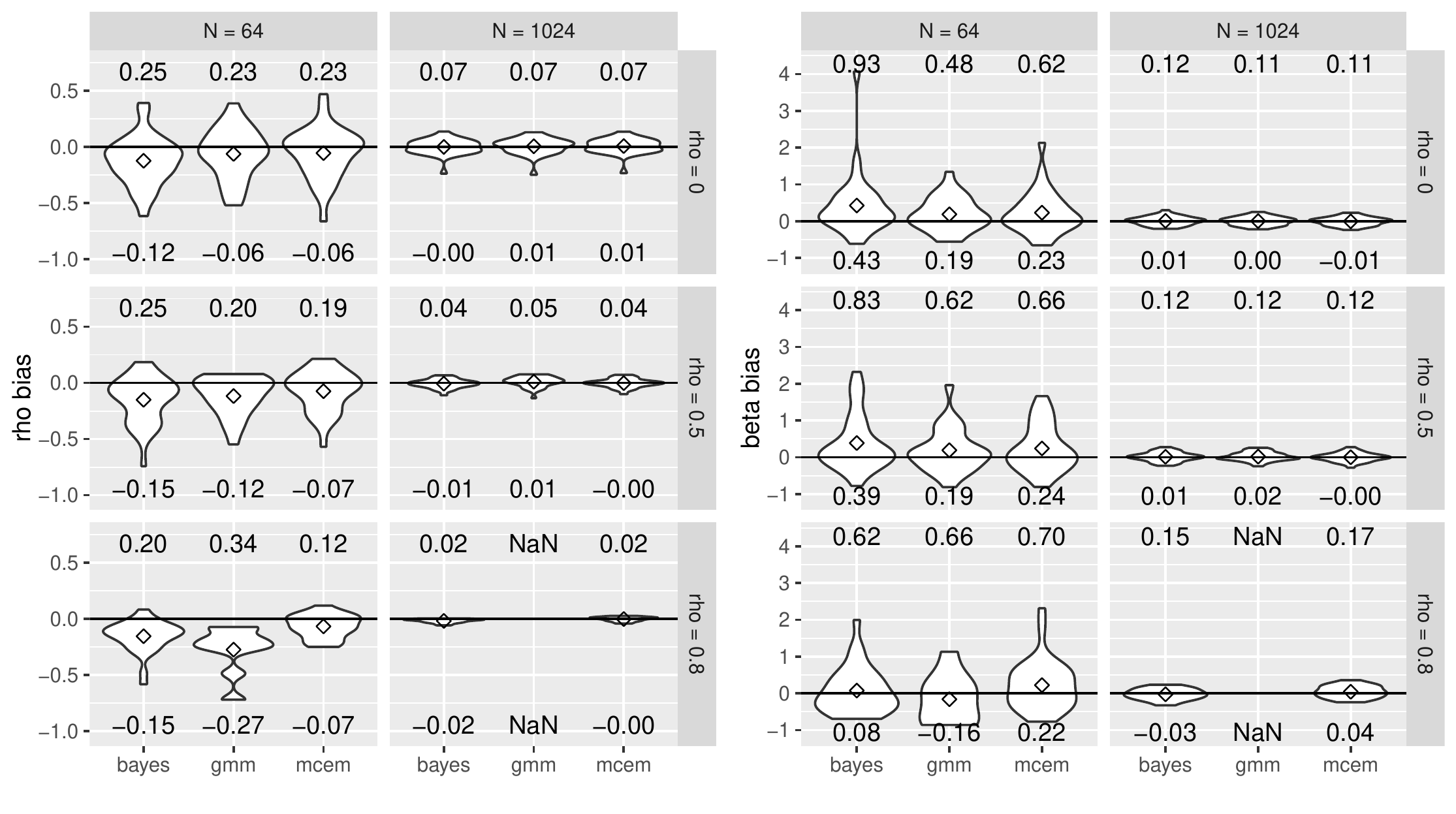}
\caption{Error distribution for the estimated $\rho$ and $\beta_1$ parameters of a Spatial Probit model, estimated from 30 simulations. Top figures in each panel indicates RMSE, bottom figures indicate bias. The $N = 1024, \ \rho = 0.8$ distribution for the \emph{GMM} model is missing because all estimation attempts failed.}
\label{fig:sprobit_bias}
\end{figure}

\begin{figure}[ht]
\includegraphics[width=\textwidth]{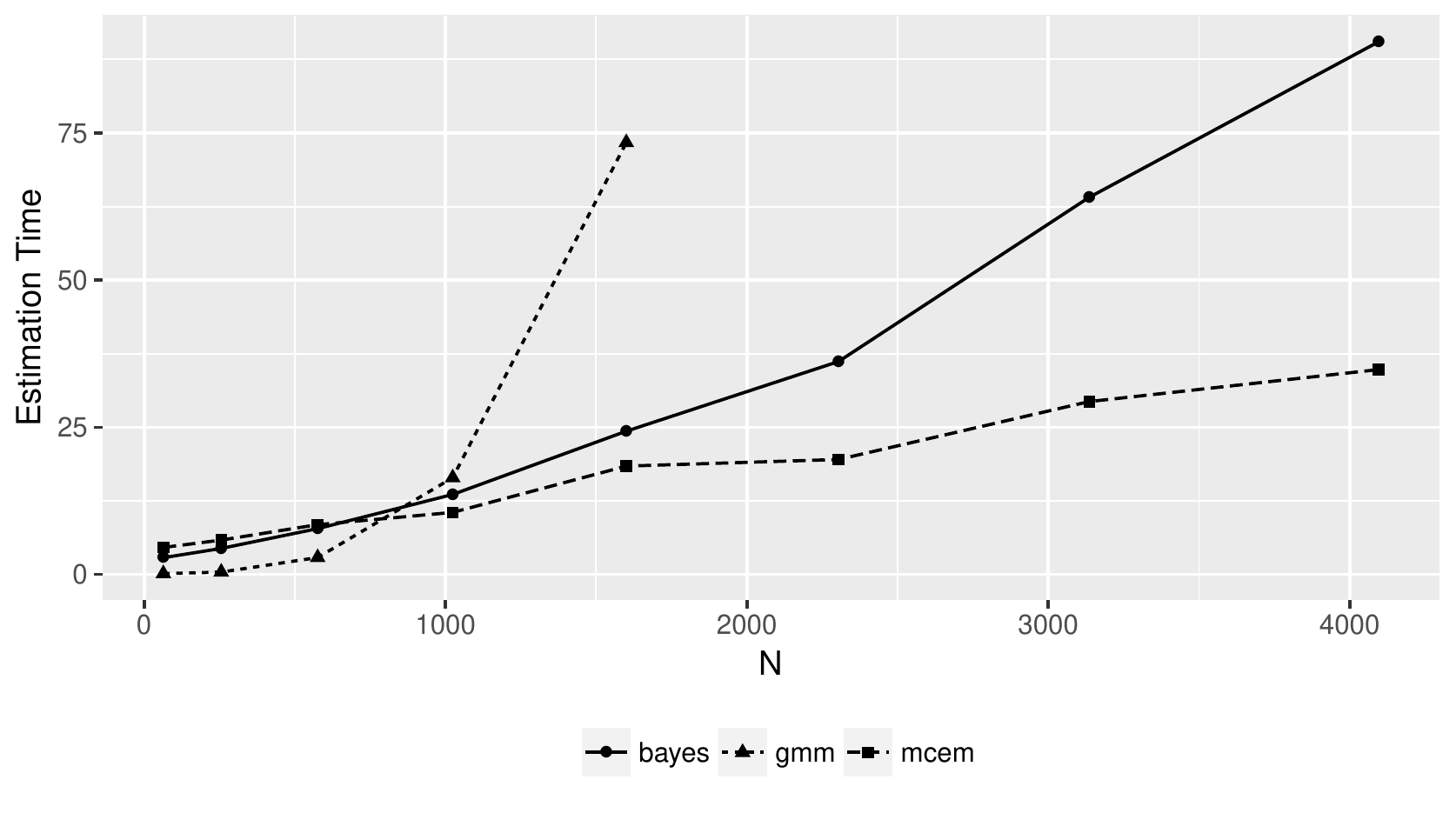}
\caption{CPU time (in seconds) for three different estimation strategies for training a Spatial Probit model, averaged over 5 simulations each.}
\label{fig:sprobit_time}
\end{figure}

In a final set of simulation experiments we compare our MCEM estimator against two alternative approaches for estimating a Spatial Probit model, i.e. a setup where the outcome is binary and the data univariate and cross-sectional (i.e. $G = T = 1$). We focus on this setting because it allows us to compare our estimator against a pair of readily implemented and widely used alternatives. The first is the Bayesian Spatial Probit model, estimated via MCMC, as proposed by \cite{lesage2000bayesian} and implemented in the R programming language by \cite{wilhelm2013estimating}. The second is the GMM estimator for Spatial Probit models, introduced by \cite{pinkse1998contracting}, and improved and implemented in R by \cite{klier2008clustering}.

We perform two sets of experiments. In the first set, we draw data from a Spatial Probit data generating process with a design matrix composed of a constant and a single predictor $x \sim N(0, 1)$, and $\beta = [0, 2]'$. The spatial weights matrix $W$ is constructed from a regular square grid with side length $N^{0.5}$. We vary the sample size ($N = {64, 1024}$), as well as the magnitude of spatial autocorrelation ($\rho = {0, 0.5, 0.8}$). We repeat each experiment 30 times.

For each simulated dataset, we use our MCEM estimator to recover point estimates for $\beta$ and $\rho$. We use a MC sample size of $S = 50$, and run the MCEM algorithm for 75 iterations, unless the maximum absolute change across all estimated parameters is smaller than $10^{-4}$ in any previous iteration. Correspondingly, we also recover point estimates using the MCMC and the GMM estimator. For MCMC sampling we rely on the default settings provided in the R package by \cite{wilhelm2013estimating}.

To establish whether the estimators are able to accurately recover the true parameters, we evaluate the distribution of the estimation errors $\rho - \hat{\rho}$ and $\beta - \hat{\beta}$. Figure~\ref{fig:sprobit_bias} plots the corresponding results. We find that (i) all three methods are unbiased and have comparable RMSE if $N$ is sufficiently high, (ii) our estimator exhibits the smallest bias/RMSE in $\rho$ if $N$ is small, and (iii) GMM estimation fails for large $N$ and $\rho$. 

In a second set of experiments, we simulate data from the same data generating process, but fix $\rho = 0.5$ and only vary $N$. The goal here is to evaluate how well the three estimators scale in terms of estimation time. The results of these experiments are reported in Figure~\ref{fig:sprobit_time}. We find that estimation time (i) increases cubically for the $GMM$ implementation, (ii) increases super-linearly for the MCMC implementation, and (iii) -- as anticipated -- increases linearly for our estimator.

\section{Application: IS Violence in Syria}
\label{sec:application}

\begin{figure}[ht]
\includegraphics[width=\textwidth]{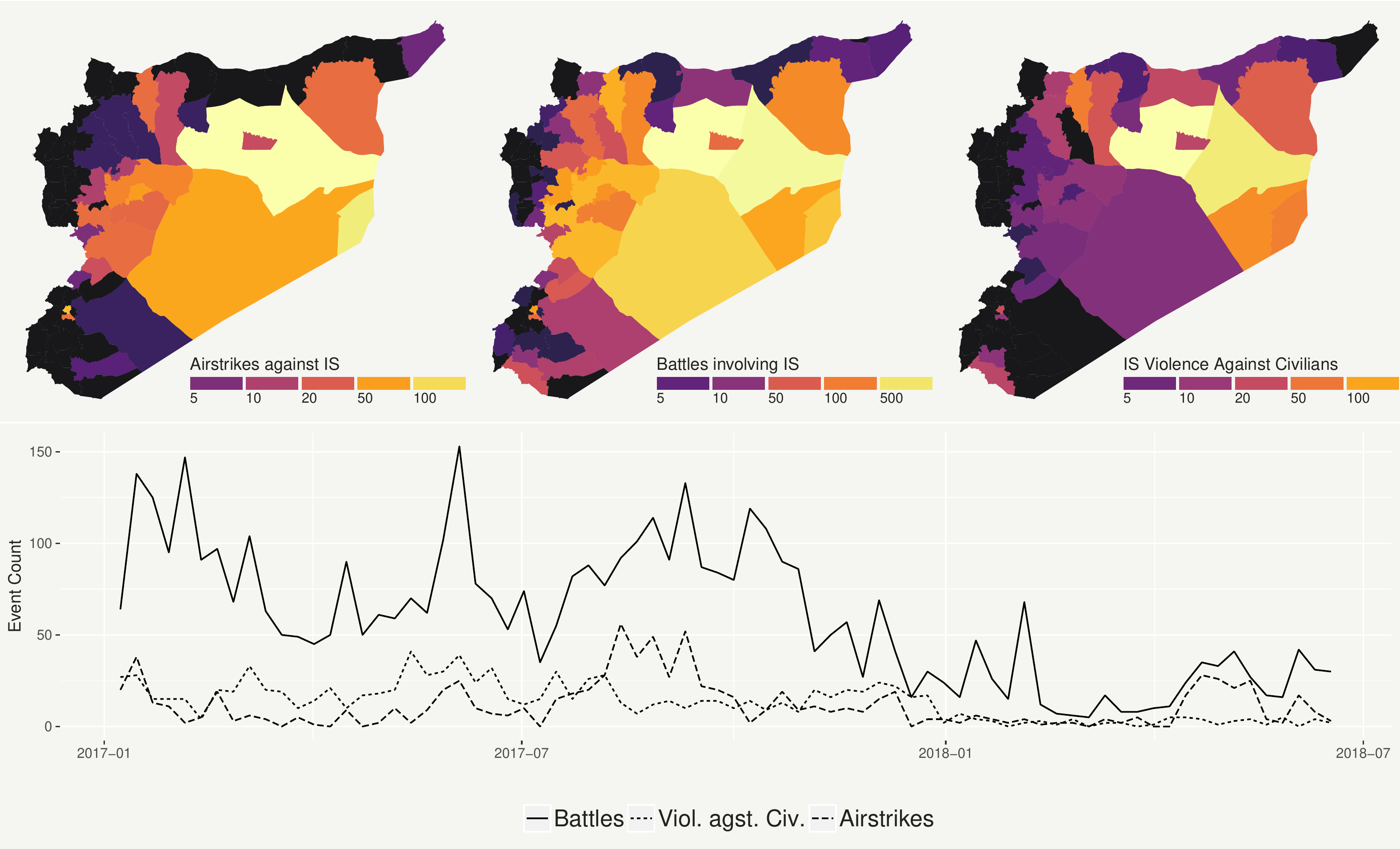}
\caption{Weekly IS-related events in Syrian Districts, January 2017 -- June 2018.}
\label{fig:events}
\end{figure}

In this section, we apply the proposed estimator to a spatio-temporal model of Islamic State (IS)-related violence in Syrian districts between January 2017 and June 2018. In particular, we investigate whether the occurrence of airstrikes targeting IS forces serves as an advance signal of IS-related violence in the following week.
Methodologically, the application demonstrates how accounting for spatio-temporal dependence may affect substantive conclusions, the use of elasticities for interpretation, and the use of out-of-sample predictions to evaluate model fit.

The application speaks to a growing number of studies in Political Science and Economics that analyze political violence using event-level data (see \citealt{schrodt2012precedents} for a review).
It is widely acknowledged that political violence exhibits considerable spatial and temporal dependence, as contested areas often shift and expand over the course of a conflict  \cite[e.g.][]{weidmann2010predicting}.
However, many analyses of political violence account for spatial and temporal dependence only via standard error adjustments, not least because appropriate modeling techniques for non-Gaussian outcomes have long been unavailable \citep[e.g.][]{pierskalla2013technology, buhaug2011s}.
Moreover, the literature suggests that there is considerable interdependence between different forms of political violence, especially during civil wars, where violence against civilians and battlefield violence between combatants are closely linked \citep{Kalyvas2006}. 
By using our estimator, we are able to fit a model that accounts for spatio-temporal dependence, and is able to capture interdependence between violence against civilians and battlefield violence directly.

We analyze a panel dataset containing weekly observations for all 68 Syrian administrative districts for the 78 calendar weeks between January 8th 2017 and June 23rd 2018, totaling 5168 observations.
We study two count outcomes: (1) The number of battle-events involving IS forces per district-week, and (2) the number of instances where IS forces perpetrated violence against unarmed civilians per district-week. Data for both variables were constructed from the ACLED dataset \citep{raleigh10acled},  which provides event-level data of political violence in Syria beginning in 2017.
We include a total of five predictors. The first four are time-invariant geographic/demographic variables that are known to correlate with the location of political violence \citep{tollefsen2015insurgency}: The area of a district; its population in 2000; the average travel time to the next city in hours, averaged over all locations in a district; and the proportion of the district area that is covered by mountains. All variables but the population count originate from the PRIO-Grid dataset \citep{tollefsen2012prio}. The population data are obtained from \cite{CIESIN2010}.
The fifth predictor is a lagged count of the number of airstrikes targeted at IS positions in a given district. This is the sole time-variant predictor, and is included under the hypothesis that airstrikes may reveal information about where future fighting is going to occur. This variable is also constructed from the ACLED dataset. All predictors except for the proportion of mountainous terrain are logged.\footnote{We add a unit constant to the airstrikes variable prior to taking the natural logarithm.}
Figure~\ref{fig:events} summarizes the spatial and temporal distribution of the two outcomes, as well as the airstrikes predictor.

We estimate three models. First, a latent-Gaussian count model without spatial or temporal dependence parameters, estimated for both outcomes separately (Model 1).
Second, a latent-Gaussian count model with spatio-temporal dependence parameters, again estimated separately for each outcome (Model 2).
And third, a multivariate latent-Gaussian count model with spatio-temporal dependence parameters, estimated for both outcomes jointly, and including a parameter for capturing outcome dependence (Model 3).
All models are estimated using our MCEM estimator, using 50 MC samples for parameter estimation and 100 MC samples for standard error estimation. Each model is run for 15 iterations, and convergence is ensured by checking the largest absolute change across all parameter estimates after each iteration, ensuring that the series is stationary.

For each model and outcome, we report point and standard error estimates, direct elasticities, and spillover elasticities. 
Direct elasticities indicate the increase in the expected outcome for unit $i$, in percent, as a result of a one percent increase in the respective predictor for unit $i$. For Model 1 the direct elasticities are equivalent to the point estimates.\footnote{Note that for the \emph{Mountains} predictor the reported value is actually a semi-elasticity because the variable is not logged.} For Models 2 and 3 the reported direct elasticities are averaged over all observations, as non-zero spatial- or outcome dependence leads to unit-specific elasticities due to the spatial/outcome-specific feedback effect.
Finally, the spillover elasticities indicate the relative increase in the sum of expected outcomes over all units $j \neq i$ in period $t$, in percent, as a result of increasing the predictor of unit $i$ in period $t$ by 1 percent. Naturally, all spillover elasticities are zero for Model 1.

\begin{sidewaystable}[htbp]\centering
\newcommand\sym[1]{\rlap{$^{#1}$}}
\footnotesize
\begin{tabular*}{\columnwidth}{
  @{\hspace{\tabcolsep}\extracolsep{\fill}}
  l*{9}{D{.}{.}{-1}}
}
  &\multicolumn{3}{c}{Model 1}&\multicolumn{3}{c}{Model 2}&\multicolumn{3}{c}{Model 3}\\
  \cmidrule(lr){2-4} \cmidrule(lr){5-7} \cmidrule(lr){8-10}
  \addlinespace
    &\multicolumn{1}{c}{Coef.}& \multicolumn{1}{c}{Direct Elast.}& \multicolumn{1}{c}{Spillover}
    &\multicolumn{1}{c}{Coef.}& \multicolumn{1}{c}{Direct Elast.}& \multicolumn{1}{c}{Spillover}
    &\multicolumn{1}{c}{Coef.}& \multicolumn{1}{c}{Direct Elast.}& \multicolumn{1}{c}{Spillover} \\
\midrule
\multicolumn{2}{l}{\bf \emph{Battles}} \\
\addlinespace
Airstrikes (log, lag) &             1.861^{***} &   1.861 &     0.000 &     0.184^{***} &   0.184 &     0.015 &     0.126^{*} &     0.140 &     0.012\\
                      &            (0.069) & & &                            (0.052) & & &                           (0.054)\\ 
Area (log)            &             2.144^{***} &   2.144 &     0.000 &     0.472^{***} &   0.472 &     0.038 &     0.389^{***} &   0.405 &     0.030\\
                      &             (0.130) & & &                           (0.088) & & &                           (0.078)\\
Population (log)      &              0.083^{*} &    0.083 &     0.000 &     0.043^{**} &    0.043 &     0.004 &     0.047^{***} &   0.048 &     0.003\\
                      &              (0.034) & & &                          (0.014) & & &                           (0.011)\\
Remoteness (log)      &              -1.154^{***} & -1.154 &    0.000 &     -0.289^{***} &  -0.290 &    -0.024 &    -0.192^{**} &   -0.198 &    -0.014\\
                      &              (0.148) & & &                          (0.071) & & &                           (0.071)\\
Mountains             &              -1.226^{***} & -1.226 &    0.000 &     -0.093 &        -0.093 &    -0.008 &    -0.034 &        -0.043 &    -0.005\\
                      &              (0.216) & & &                          (0.051) & & &                           (0.047) \\
$\gamma$              &              & & &                                  0.866^{***} & & &                       0.874^{***}\\
                      &              & & &                                  (0.016) & & &                           (0.011)\\
$\rho$                &              & & &                                  0.076^{***} & & &                       0.063^{***}\\
                      &              & & &                                  (0.014) & & &                           (0.011)\\
\midrule
\multicolumn{2}{l}{\bf \emph{Violence Against Civilians}} \\
\addlinespace
Airstrikes (log, lag) &             0.954^{***} &   0.954 &     0.000 &     0.404^{***} &   0.407 &     0.086 &     0.313^{***} &   0.321 &    0.059 \\
                      &            (0.045) & & &                            (0.052) & & &                           (0.050)\\
Area (log)            &              1.251^{***} &  1.251 &     0.000 &     0.454^{***} &   0.457 &     0.097 &     0.319^{***} &   0.339 &    0.063 \\
                      &              (0.136) & & &                          (0.066) & & &                           (0.072)\\
Population (log)      &              0.072^{***} &  0.072 &     0.000 &     0.030 &         0.030 &     0.006 &     0.031^{*} &     0.033 &     0.006\\
                      &              (0.017) & & &                          (0.016) & & &                           (0.013)   \\
Remoteness (log)      &              -0.343^{***} & -0.343 &    0.000 &     -0.241^{***} &  -0.242 &    -0.051 &    -0.104 &        -0.113 &    -0.021\\
                      &              (0.079) & & &                          (0.067) & & &                           (0.066)  \\
Mountains             &              -0.722^{***} & -0.722 &    0.000 &     -0.186^{**} &   -0.187 &    -0.040 &    -0.190^{***} &  -0.193 &   -0.035\\
                      &              (0.092) & & &                          (0.055) & & &                           (0.053)\\
$\gamma$              &              & & &                                  0.679^{***} & & &                       0.643^{***}\\
                      &              & & &                                  (0.019) & & &                           (0.018)\\
$\rho$                &              & & &                                  0.181^{***} & & &                       0.158^{***}\\
                      &              & & &                                  (0.017) & & &                           (0.027)\\
\midrule
$\lambda$             &              & & &   & & &                                                                  0.045^{***} \\
                      &              & & &   & & &                                                                  (0.009) \\
\midrule
Obs. & \multicolumn{1}{c}{5168} & & & \multicolumn{1}{c}{5168} & & &  \multicolumn{1}{c}{5168} \\
\bottomrule
\multicolumn{10}{l}{\footnotesize * \(p<0.05\), ** \(p<0.01\), *** \(p<0.001\)}\\
\multicolumn{10}{l}{\footnotesize Note: Standard errors in parentheses. All elasticities are averaged over units and time-periods.}\\
\end{tabular*}
\vspace{0.2cm}
\caption{Spatio-temporal models of weekly IS-related battles \& violence against civilians in Syrian districts, Jan. 2017 -- Jun. 2018.\label{tab1}}
\end{sidewaystable}

\begin{table}[htbp]\centering
\newcommand\sym[1]{\rlap{$^{#1}$}}
\footnotesize
\begin{tabular*}{\columnwidth}{
  @{\hspace{\tabcolsep}\extracolsep{\fill}}
  l*{6}{D{.}{.}{-1}}
}
  &\multicolumn{2}{c}{Model 1}&\multicolumn{2}{c}{Model 2}&\multicolumn{2}{c}{Model 3}\\
  \cmidrule(lr){2-3} \cmidrule(lr){4-5} \cmidrule(lr){6-7}
  \addlinespace
    &\multicolumn{1}{c}{RMSE}& \multicolumn{1}{c}{MAE}
    &\multicolumn{1}{c}{RMSE}& \multicolumn{1}{c}{MAE}
    &\multicolumn{1}{c}{RMSE}& \multicolumn{1}{c}{MAE} \\
\midrule
\emph{Battles} &                13.930 &    1.835 & 2.599^{*} & 0.680^{*} & 2.677 & 0.699\\
\emph{Viol. agst. Civilians} &  1.094 &     0.413 & 0.946 & 0.309 & 0.934^{*} & 0.298^{*}\\
\bottomrule
\multicolumn{7}{l}{* Smallest value per outcome.}\\
\end{tabular*}
\vspace{0.2cm}
\caption{Out-of-sample loss for all models summarized in Table \ref{tab1}. RMSE: Root mean squared error; MAE: Mean absolute error.\label{tab2}}
\end{table}

Table \ref{tab1} summarizes the results. Three findings are worth highlighting.
First, airstrikes appear to be a useful predictor of both battle events as well as violence against civilians. However, the effect size decreases drastically once we take spatio-temporal dependence into account, namely from a direct elasticity of 1.86/0.95 to 0.18/0.41. This is likely because Model 1 falsely attributes temporal autocorrelation in the outcomes to the airstrikes variable. Still, the predictive power of airstrikes is substantial. For instance, a doubling of the number of airstrikes observed in a given district is expected to increase the number of events involving violence against civilians in the same district by between 32\% (Model 2) and 41\% (Model 3) in the following week. 
Second, we find little evidence of spatial dependence in battles, but some evidence for spatial autocorrelation for the violence against civilians outcome. For instance, a doubling of the number of airstrikes in a district leads to between 8.6\% (Model 2) and 5.9\% (Model 3) more instances of violence against civilians across all other districts in the following week. 
Third, surprisingly, we find only very moderate dependence between the two outcomes, with the $\lambda$ estimate in Model 3 statistically significant, but close to zero.

Finally, we compare the models' goodness of fit via their out-of-sample predictive power.
Specifically, we reestimate all models while censoring the outcome value for an arbitrary 33\% of all observations. Thanks to the MC step of our MCEM estimator, predictions for the omitted outcomes are generated automatically as a side-product of the estimation procedure.
Table~\ref{tab2} reports root mean squared prediction error (RMSE) and mean absolute prediction error (MAE) for all models. Clearly, accounting for spatio-temporal dependence leads to much better model fit, with both error metrics being smaller in Models 2 and 3 than in Model 1. 
Moreover, the multivariate model (Model 3) yields a slightly better fit for the violence-against-civilians outcome, but not for the battle outcome.

\section{Concluding Remarks}
\label{sec:conclusion}

The goal of the work reported here was to introduce a scalable method for estimating models of spatio-temporal dependence in non-Gaussian and potentially multivariate outcomes. 
We addressed this task by proposing a Monte Carlo Expectation Maximization (MCEM) procedure for estimating the parameters of a family of latent-Gaussian spatio-temporal autoregressive models, discussing the cases of binary and count outcomes explicitly.
Relying on an MCEM approach avoids having to evaluate the intractable likelihood of the latent Gaussian spatio-temporal model, and offers a straightforward strategy for generating out-of-sample predictions during model training.

We confirmed through simulation experiments that our estimator yields consistent parameter and standard error estimates in finite samples. 
We also showed that our proposed estimator scales well with large datasets, both analytically and in simulations. In the case of a univariate outcome and a fixed number of spatial neighbors per unit, our implementation of the expected log-likelihood function has linear time complexity in the number of units and time periods. For multivariate outcomes, the time complexity is approximately quadratic in the number of units. Further, the simulation experiments in Section~\ref{sec:spat_probit} highlighted that our estimator compares favorably against other options for estimating the Spatial Probit model, both in terms of error and estimation time. 

Finally, we illustrated the usefulness of our estimator by examining weekly counts of IS-related violence in Syria. We demonstrated the use of elasticities for interpretation, and how out-of-sample predictions may be employed to assess model fit. Substantively, we found that airstrikes may serve as a useful advance signal for predicting the occurrence of civil-war related violence.

The work presented here may be usefully extended in several directions. While we limit our discussion to binary and conditionally Poisson distributed count outcome data, there are various other outcome distributions that may be of interest. In particular, using the Negative Binomial may be necessary for count outcomes that are highly overdispersed, and thus inadequately modeled by the Poisson log-Normal setup discussed in this paper \citep[see also][]{liesenfeld2016likelihood}.
Another potentially fruitful avenue for further research may be extending our MCEM procedure such that the MC sample size is determined dynamically during estimation, e.g. using the method introduced by \cite{caffo2015}. This would likely lead to estimates exhibiting less variance, potentially speed up estimation, and provide a more appropriate strategy for assessing convergence. 

\FloatBarrier
\newpage

\section*{Acknowledgments}
We would like to thank Frederick Boehmke, Lars-Erik Cederman, Robert Franzese Jr., Dennis Quinn, Andrea Ruggeri, Emily Schilling, Martin Steinwand, Oliver Westerwinter, Michael Ward, Christopher Zorn, and Bruce Desmarais for their useful comments on previous versions of this paper.

\bibliography{main}

\bibliographystyle{apsr}

\FloatBarrier
\newpage

\section*{Appendix}

\subsection*{Bias, RMSE, and Coverage for the Spatio-Temporal Probit Model}

\begin{figure}[ht]
\includegraphics[width=\textwidth]{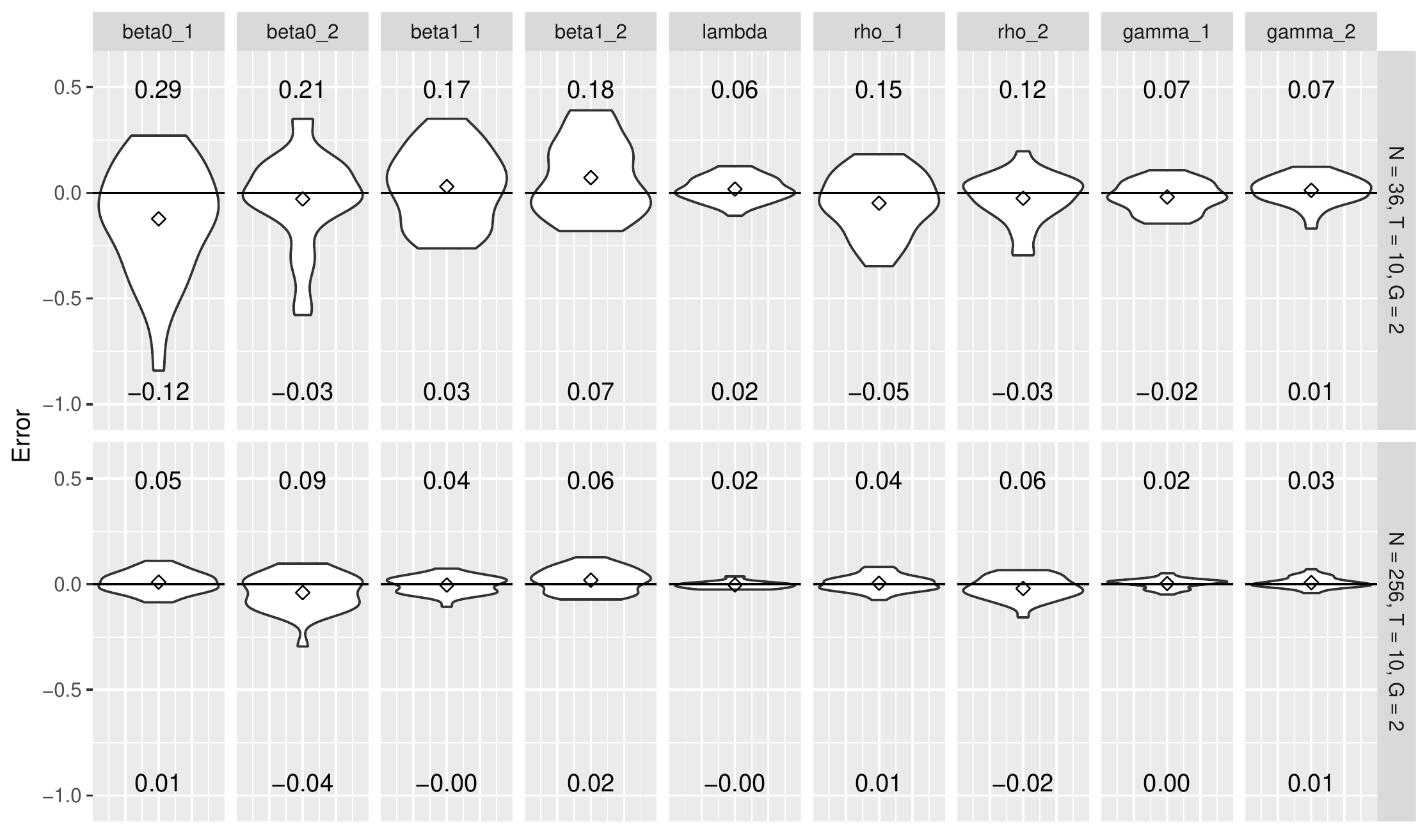}
\caption{Error distribution for all parameter estimates in a spatio-temporal Probit model with a bivariate outcome, for two different sample sizes. Estimates based on 50 Monte Carlo simulations per sample size. Top figure in each panel indicates RMSE, bottom figures indicate bias.}
\label{fig:bias_binary}
\end{figure}

\begin{figure}[ht]
\includegraphics[width=\textwidth]{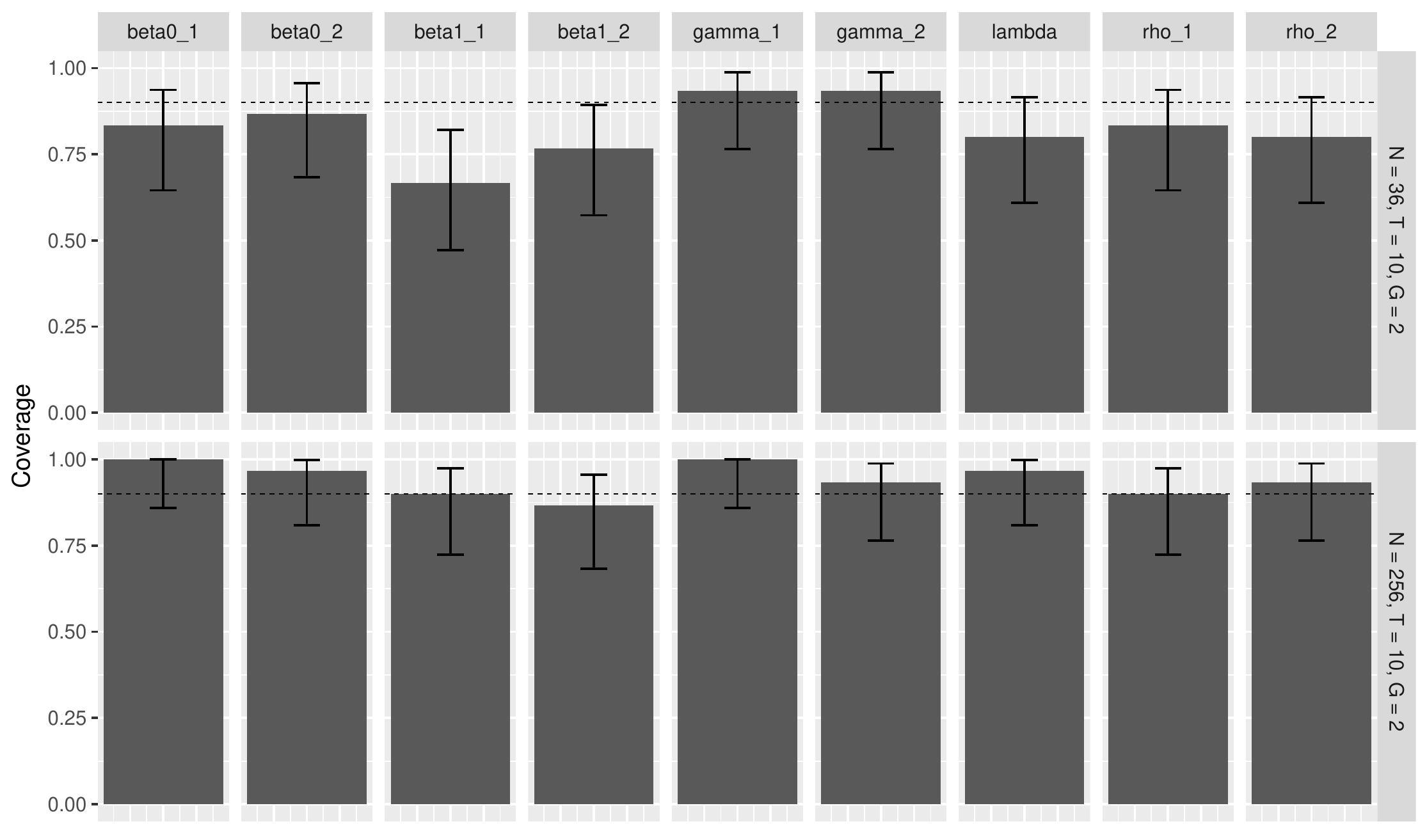}
\caption{90\% CI coverage probability for all parameter estimates in a spatio-temporal Probit model with a bivariate outcome. Estimates based on 50 Monte Carlo simulations per sample size. Error bars represent 95\% CI for the coverage proportion.}
\label{fig:coverage_binary}
\end{figure}

\end{document}